\documentclass[letterpaper]{article} 
\usepackage{aaai25}  
\usepackage{times}  
\usepackage{helvet}  
\usepackage{courier}  
\usepackage[hyphens]{url}  
\usepackage{graphicx} 
\usepackage{natbib}  
\usepackage{caption} 
\usepackage{algorithm}
\usepackage{amsthm}
\usepackage{amsmath}
\usepackage{multirow}
\usepackage{float}
\usepackage{algorithm}
\usepackage{algpseudocode}
\usepackage{bm}
\usepackage{graphicx}
\usepackage{caption}
\newtheorem{theorem}{Theorem}
\usepackage{subfigure}
\usepackage{adjustbox}
\usepackage[utf8]{inputenc} 
\usepackage[T1]{fontenc}    
\usepackage{url}            
\usepackage{booktabs}       
\usepackage{amsfonts}       
\usepackage{nicefrac}       
\usepackage{microtype}      
\usepackage{xcolor}         
\usepackage{newfloat}
\usepackage{listings}
\DeclareCaptionStyle{ruled}{labelfont=normalfont,labelsep=colon,strut=off} 
\floatstyle{ruled}
\newfloat{listing}{tb}{lst}{}
\floatname{listing}{Listing}
\title{From Pairwise to Ranking: Climbing the Ladder to Ideal Collaborative Filtering with Pseudo-Ranking}
\author{
    Yuhan Zhao\textsuperscript{\rm 1, 2}, Rui Chen\textsuperscript{\rm 1}\thanks{Corresponding authors}, Li Chen\textsuperscript{\rm 2}, Shuang Zhang\textsuperscript{\rm 1}, Qilong Han\textsuperscript{\rm 1}, Hongtao Song\textsuperscript{\rm 1}\footnotemark[1]\\
}
\affiliations{
    \textsuperscript{1} College of Computer Science and Technology, Harbin Engineering University \\
    \textsuperscript{2} Department of Computer Science, Hong Kong Baptist University\\
    \{csyhzhao,lichen\}@comp.hkbu.edu.hk, \{ruichen, zhangshuang, hanqilong, songhongtao\}@hrbeu.edu.cn
}

\usepackage{bibentry}

\begin{document}

\maketitle
\begin{abstract}
Intuitively, an ideal collaborative filtering (CF) model should learn from users' full rankings over all items to make optimal top-K recommendations. Due to the absence of such full rankings in practice, most CF models rely on pairwise loss functions to approximate full rankings, resulting in an immense performance gap. In this paper, we provide a novel analysis using the multiple ordinal classification concept to reveal the inevitable gap between a pairwise approximation and the ideal case. However, bridging the gap in practice encounters two formidable challenges: (1) none of the real-world datasets contains full ranking information; (2) there does not exist a loss function that is capable of consuming ranking information. To overcome these challenges, we propose a pseudo-ranking paradigm (PRP) that addresses the lack of ranking information by introducing pseudo-rankings supervised by an original noise injection mechanism. Additionally, we put forward a new ranking loss function designed to handle ranking information effectively. To ensure our method's robustness against potential inaccuracies in pseudo-rankings, we equip the ranking loss function with a gradient-based confidence mechanism to detect and mitigate abnormal gradients. Extensive experiments on four real-world datasets demonstrate that PRP significantly outperforms state-of-the-art methods. 

\end{abstract}

\section{Introduction}
Collaborative filtering (CF) has been one of the most fundamental techniques in recommender systems due to its simplicity and effectiveness~\cite{MZW21}. It leverages user-item interactions to learn user preferences and make top-K recommendations~\cite{HLZ17,WHW19,HDW20}. To find the top-K items, a CF model normally needs to generate a full ranking over the entire item universe. To fulfill this task, ideally, given a set of users $U$, one should derive a CF model that maximizes the posterior probability $\prod_ {u \in U}p(\Theta \mid \Psi_u)$, where $\Theta$ denotes the set of the model's parameters that need to learn, and $\Psi_u$ is user $u \in U$'s full ranking of all items. However, this ideal optimization goal is unattainable in real-world scenarios, as it is impractical for users to fully rank a vast number of items.

In practice, most CF methods~\cite{HDW20,WWF21,YYX22} rely on pairwise loss functions to approximate full rankings. We refer to these as pairwise CF. Its core idea is to make positive item $i_p$ (i.e., an item interacted by a user) more similar to user $u$ than negative item $i_n$. For example, Bayesian personalized ranking (BPR)~\cite{RFG12} pairs each positive item with a negative item to form a pairwise relationship for training. Formally, a pairwise loss function forms a pairwise relationship $i_p >_{u} i_n$ (i.e., $u$ prefers item $i_p$ to $i_n$) as the ground truth. In this context, the optimization objective is to maximize the posterior probability $\prod_ {(u, i_p, i_n) \in D}p(\Theta \mid i_p>_{u} i_n)$ based on a set of pairwise relationships, where $D$ is the training dataset. Despite the strong performance of pairwise CF models, approximating full rankings with pairwise relationships introduces a significant gap compared to the ideal CF model. In this paper, we provide a novel analysis using the concept of multiple ordinal classification to \textit{reveal the inherent gap between the pairwise approximation and the ideal case}. However, bridging the gap in practice encounters two formidable challenges:
\begin{itemize}
    \item \textbf{None of the real-world datasets contains full ranking information}: Obtaining full rankings requires extensive user participation, a feat that is unachievable in most practical scenarios.
    \item \textbf{Lack of a loss function that is capable of handling ranking information}: Current CF loss functions focus on non-ranking relationships, such as pairwise or listwise~\cite{XLW08}, and are not designed to effectively and efficiently process ranking data.
\end{itemize}
To overcome the aforementioned challenges, we introduce a novel pseudo-ranking paradigm (PRP). For the first challenge, we propose a ranker module to generate a pseudo-ranking of given items. However, training the ranker demands supervised signal (i.e. ranking information), presenting a classic chicken-and-egg dilemma. We innovatively propose a noise injection mechanism to add different levels of noise to positive samples to construct a ranking in line with user preferences. The intuition is that when the injection noise is very small, the semantics of the constructed $i_p^2$ will not change significantly. When the injection noise is large, the semantics of the positive item $i_p^3$ will be seriously broken~\cite{HHD18,ZPL23}, and users are difficult to love the item. Accordingly, $i_p >_u i_p^2 >_u i_p^3$ is constructed as a supervision. For the second challenge, inspired by the multiple ordinal classification concept, we then put forward a ranking loss function grounded in classification principles. This function ensures that an item's score is proportional to its ranking order: the higher the rank, the higher the score. To further ensure our method's robustness against potential inaccuracies in pseudo-rankings provided by the ranker, we equip the ranking loss function with an original gradient-based confidence mechanism, which detects outliers by assessing gradient density. By reducing the weights of outliers, the impact of inaccurate ranking information on training can be alleviated. We summarize our main contributions as follows:
\begin{itemize}
    \item To the best of our knowledge, we are the first to use the multiple ordinal classification concept to approximate ideal CF rather than traditional pairwise loss. We provide a new direction for future development of loss functions.
    \item We propose a novel ranker, which can generate pseudo-rankers according to user preferences through a noise injection mechanism, effectively solving the problem that the existing dataset does not contain ranking information.
    \item We propose a novel loss fucnion, which can directly deal with ranking data that traditional loss fucnions cannot directly deal with, and a gradient-based confidence mechanism, which can judge outliers by gradients and alleviate inaccurate ranking information.
    \item We conduct extensive experiments on four real-world datasets that confirm that PRP can achieve significant improvements over representative state-of-the-art methods. Moreover, when combined with PRP, a wide variety of mainstream CF models can consistently and substantially boost their performance. 
\end{itemize}

\section{Understanding Ideal Collaborative Filtering }
\begin{figure*}[t]
  \centering
\includegraphics[width=0.8\textwidth]{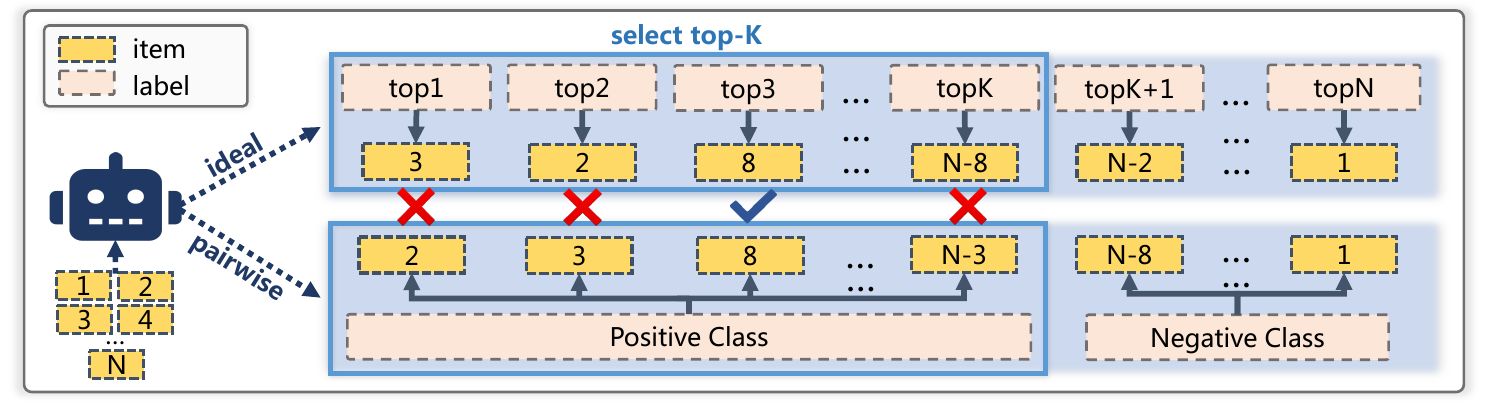}
  \caption{A comparison between ideal CF and pairwise CF.}
  \label{fig:nips}
\end{figure*}
According to maximum likelihood estimation, an ideal CF model should be trained by maximizing the posterior probability $\prod_ {u \in U}p(\Theta \mid \Psi_u)$, where $\Psi_u$ represents the full ranking of all items by user $u \in U$. However, in real-world scenarios, this ideal approach is impractical due to the absence of complete ranking information in datasets. Consequently, conventional CF models approximate this by using pairwise loss functions, shifting the optimization objective to maximizing the posterior probability $\prod_ {(u, i_p, i_n) \in D}p(\Theta \mid i_p>_{u} i_n)$. A classic example of this is BPR:
\begin{equation}
\mathcal{L}_{BPR}= -\ln [\sigma(s_u(i_p) - s_u(i_n))],
\label{eq:bpr}
\end{equation}
where $s_u(\cdot)$ denotes the score of an item for user $u$, and $\sigma(\cdot)$ is the $\operatorname{sigmoid}$ function. Models trained with pairwise loss functions are termed pairwise CF models. 

To model the ideal CF, inspired by works such as~\cite{YK20,LWB07}, we propose a novel approach: transforming CF ranking into a multiple ordinal classification problem. Specifically, we hypothesize the existence of $N$ distinct labels $(top1, top2, top3, \ldots, topN)$. For each user $u$, the ideal CF model would assign each item to the appropriate label based on user preference. When making recommendations, these labels are treated as ordinal, with lower $top*$ values indicating higher recommendation priority. Figure~\ref{fig:nips} illustrates this concept. The classification loss in this context can be expressed as:
\begin{equation}
\mathcal{L}_{class} = -\log\frac{\exp(z_t)}{\sum_{v=1}^{N} \exp(z_v)},
\end{equation}
where $z_t$ denotes the score of the target class. Under this framework, pairwise CF can be viewed as a process where each user $u$ considers certain labels (e.g., $top1, top2, \ldots, topK$) as positive samples, while the remaining labels are negative. This comparison highlights that existing pairwise CF methods fall short of achieving ideal CF, exposing a significant gap. This observation raises a critical question: \textbf{Is there a connection between optimizing ideal CF and existing metrics (e.g., NDCG)}? While these metrics are not without flaws, they have proven effective in both practical and theoretical settings. Therefore, it is neither practical nor advisable to disregard them entirely~\cite{PCH24}. Understanding the relationship between ideal CF and traditional metrics is essential.
\begin{theorem}
\label{th:2}
Optimizing ideal CF is a sufficient but not necessary condition for optimizing NDCG.
\end{theorem}

\begin{proof} 
For convenience, we analyze discounted cumulative gain (DCG, with NDCG being its normalized version). DCG is defined as:
\begin{equation}
    DCG_{K} = \sum^{K}_{j=1} \frac{\mathbb{I}(i_{j} \in \mathcal{P})}{\log(j+1)}, 
\end{equation}
where $\mathbb{I}(\cdot)$ is an indicator function, and $\mathcal{P}$ represents the ground truth set of positive items.

\textbf{Sufficient Condition (if part):} If $\Psi_u$ achieves ideal CF, it will rank all positive items at the top, thereby satisfying:
\begin{equation}
\{\Psi_u(1), \Psi_u(2), \ldots, \Psi_u(K)\} = \mathcal{P}.
\end{equation}
This implies that ideal CF $\implies \max DCG_{K}$.

\textbf{Necessary Condition (only if part):}
We can construct a different ranking $\Psi_u^{\prime}$ such that:
\begin{equation}
\forall (j<K), \Psi_u^{\prime}(j) = \Psi_u(j+1), \Psi_u^{\prime}(K)= \Psi_u(1).
\end{equation}
Clearly, ${\Psi_u^{\prime}(1),\Psi_u^{\prime}(2),\Psi_u^{\prime}(3),\ldots \Psi_u^{\prime}(K)} = \mathcal{P}$, satisfying the maximum DCG${K}$. However, $\Psi_u \neq \Psi_u^{\prime}$, meaning that maximizing DCG${K}$ does not necessarily imply ideal CF.
\end{proof}

\section{Methodology}
In this section, we introduce a concrete implementation of the PRP. Note that these steps can be implemented by different methods and thus the overall paradigm is generic. The implementation hinges on two fundamental components: a ranker for generating pseudo-rankings and a ranking loss function tailored to effectively leverage the ranking information.

\subsection{Ranker}
The absence of ranking information in the dataset presents a significant challenge. To address this, we introduce a ranker module capable of generating accurate pseudo-rankings to facilitate training. The items to be ranked are represented as $M = \{i_1, i_2, i_3, \dots, i_k\}$, which are randomly sampled from the dataset. The ranker aims to learn a user-specific mapping $\gamma_{u}: M \rightarrow \pi_{u}$, where $\pi_{u}$ denotes the ranking result. Specifically, $\pi_{u}(v)$ represents the item at position $v$ in the ranking, and $\gamma_{u}(i)$ denotes the position of item $i$ in $\pi_{u}$. This mapping is required to satisfy the condition: $\forall i_v, i_w \in M, \text{ if } \gamma_{u}(i_v) < \gamma_{u}(i_w), \text{ then } i_v >_u i_w \text{ and } s_u(i_v) > s_u(i_w)$. We also impose an additional constraint aligned with a common assumption in CF: interacted items should receive higher scores than non-interacted ones~\cite{PC13,RFG12, MZW21}.

To realize this objective, an intuitive approach is to design a neural network that takes user $u$ and set $M$ as input, outputting scores for all items to enable ranking:
\begin{equation}
\begin{aligned}
s_u(i_1),s_u(i_2),\ldots, s_u(i_k) &= Rank_{u}(\mathbf{e}_{u},\mathbf{e}_{1},\mathbf{e}_{2},...,\mathbf{e}_{k}), \\
\pi_{u} = argsort(s_u(i_1),&s_u(i_2),\ldots, s_u(i_k)).
\label{eq:Prank}
\end{aligned}
\end{equation}
The $Rank_{u}$ module can employ any network architecture. However, for efficiency, we use a simple multi-layer perceptron (MLP) in this work. Despite its simplicity, this approach lacks a specific training objective, making it insufficient to guarantee the acquisition of reliable ranking knowledge. Consequently, there is no direct assurance that the network's output will correspond to the desired ranking outcome. Therefore, an additional training objective is required to guide the network's learning process. However, the lack of ranking information in the dataset creates a classic chicken-and-egg dilemma.

To circumvent this issue, we propose a noise injection mechanism to construct ranking signals based on interacted items, denoted as $\mathbf{e}_{p}^1$. The core intuition is that injecting a small amount of noise into an item $\mathbf{e}_{p}$ produces a new item $\mathbf{e}_{p}^2$ with slightly altered semantics, reflecting a slight decrease in user preference. Given the substantial noise magnitude, resulting in $\mathbf{e}_{p}^3$, the item's semantics degrade significantly, leading to a further reduction in user preference~\cite{HHD18,ZPL23}. We formalize this by applying two different magnitudes of noise to the original item $\mathbf{e}_{p}$, generating a set $\mathcal{P}_{n}$:

\begin{equation}
\begin{split}
\begin{aligned}
&\mathcal{P}_{n} = \{ i_{p}^{m} | \mathbf{e}_{p}^{m}=\mathcal{T}_{\theta_m}(\mathbf{e}_{p}),m=1,2,3\},\\
&\mathcal{T}_{\theta_m}(\mathbf{e}_{p}) = \mathbf{e}_{p} + \theta_{m} \epsilon_{u}, \\
\end{aligned}
\end{split}
\end{equation}
$\mathcal{T}_{\theta_m}$ represents the noise injection function, $\epsilon_{u}$ denotes the noise, and $\theta_m$ denotes the noise magnitude. We set $\theta_1 = 0, \theta_3 \gg \theta_2$. While random noise injection is a straightforward idea, it is unreliable. If a user’s preferences are robust, significant noise is necessary to alter their ranking preferences. Consequently, the noise distribution should be user-dependent. Assuming no additional prior knowledge, we model the noise as following a Gaussian distribution, with the parameters generated as follows:
\begin{align}
\mu_u &= MLP_1(\mathbf{e}_{u}),\\
\log \sigma_u^2 &= MLP_2(\mu_i).
\end{align}
We opt to fit $\log \sigma_u^2$ instead of $\sigma_u^2$ directly due to the non-negativity constraint on $\sigma_u^2$, which would otherwise require an activation function~\cite{YZL23}. Once we obtain the mean $\mu_u$ and variance $\sigma_u^2$, the latent noise $\epsilon_{u}$ is generated by sampling from $\mathcal{N}(\mu_u, \sigma_u^2)$. Direct optimization of this process is infeasible due to its non-differentiable nature. To address this, we employ the reparameterization trick~\cite{CLJH22}:
\begin{align}
\epsilon_{u} &= \mu_u + \sigma_u \cdot \eta,\\
\eta &\sim \mathcal{N}(0, \mathbf{I}),
\end{align}
where $\eta \sim \mathcal{N}(0, \mathbf{I})$, with $\mathbf{I}$ as the identity matrix, is Gaussian noise. The resulting ranking is then processed by our ranking loss function, denoted as $\mathcal{L}_{p}$, to supervise the ranker’s training:
\begin{equation}
\begin{split}
\begin{aligned}
\mathcal{L}_{p} = \mathcal{L}_{rank}^{\alpha}(u,\gamma_{u}(\mathcal{P}_{n})).
\end{aligned}
\end{split}
\end{equation}

\subsection{Ranking Loss Function}
With the ranker generating pseudo-ranking data, the challenge shifts to designing an appropriate loss function. Existing loss functions, typically following a pairwise paradigm, are inadequate for capturing complex ranking information. Our ideal CF theory suggests that recommendation tasks can be viewed as multiple ordinal classification problems. While a classification loss might seem applicable, the absence of explicit labels such as $(top1, top2, top3, \ldots, topN)$ renders this approach unfeasible. Nevertheless, the classification perspective inspires the design of our loss function. In classification tasks, the classical objective is to output scores such that the score for the target class exceeds that for non-target classes:
\begin{equation}
\max \sum_{v=1}^{k} \sum_{w=1, w\neq v}^{k} \max(z(v) - z(w) , 0).
\end{equation}
This formulation, however, only addresses the classification task and fails to account for the ordinal nature of ranking. Moreover, without label information, computing $z(\cdot)$ is impractical. To incorporate ordinal properties, we reformulate the problem to ensure that $\forall v < w, s_u(\pi_{u}(v)) > s_u(\pi_{u}(w))$:
\begin{align}
\mathcal{L}_{rank} &= \sum_{v=1}^{k} \sum_{w=v+1}^{k} \max(s_u(\pi_{u}(w)) - s_u(\pi_{u}(v)), 0)
\label{eq:rank_all}
\\
&= \sum_{v=1}^{k} \max(\max_{w \neq v}\{s_u(\pi_{u}(w))\} - s_u(\pi_{u}(v)), 0).
\label{eq:rank_max}
\end{align}
Equation~\ref{eq:rank_all} effectively decomposes the ranking problem into multiple sub-ranking problems. For example, if the ranking is $(i_3, i_1, i_2)$, the scores should satisfy the sub-rankings $s_u(i_3) > s_u(i_1) > s_u(i_2)$ and $s_u(i_1) > s_u(i_2)$. By exploiting the principle of transitivity, we establish that if the conditions $s_u(i_3) > s_u(i_1)$ and $s_u(i_1) > s_u(i_2)$ hold, we can directly infer the ordering $s_u(i_3) > s_u(i_1) > s_u(i_2)$, which ultimately leads to Equation~\ref{eq:rank_max}. However, the non-differentiability of the max function renders it unsuitable as a loss function, necessitating further modification:
\begin{align}
\mathcal{L}_{rank} &\approx \sum_{v=1}^{k} \log(1+\exp(\max_{w \neq v}\{s_u(\pi_{u}(w))\} - s_u(\pi_{u}(v))) \\ 
&\approx \sum_{v=1}^{k-1} \log(1+\exp(s_u(\pi_{u}(v+1)) - s_u(\pi_{u}(v))).
\end{align}
In the derivation above, we first replace the non-differentiable $max(\cdot,0)$ with $\operatorname{softplus}$ function $log(1+\exp(\cdot))$. Thanks to the ranker's result, we can then simply use $v+1$ to represent $\max_{w \neq v}\{s_u(\pi_{u}(w))\}$. This substitution eliminates the need for significant additional computational effort to find $\max_{w \neq v}\{s_u(\pi_{u}(w))\}$, thereby significantly enhancing the algorithm's efficiency. For instance, in the ranking $(i_3, i_1, i_2)$ ensuring $s_u(i_3) > s_u(i_1)$ and $s_u(i_1) > s_u(i_2)$ naturally constructs $s_u(i_3) > s_u(i_1) > s_u(i_2)$.

Given that the pseudo-ranking from the ranker may not always be accurate, we introduce a confidence mechanism. Specifically, we introduce a confidence coefficient $\alpha_v$. If a ranking is deemed inaccurate, we reduce its corresponding coefficient to lessen its impact on the overall loss:
\begin{align}
\mathcal{L}_{rank}^{\alpha} = \sum_{v=1}^{k-1} \alpha_{v} \log(1+\exp(s_u(\pi_{u}(v+1)) - s_u(\pi_{u}(v))).
\end{align}
An intuitive approach to designing $\alpha_{v}$ would be to set it as a hyperparameter, but this would require extensive expert knowledge and dataset-specific tuning. Inspired by previous works~\cite{GDH22,LLW19,WWS22}, we propose learning $\alpha_v$ based on gradient. If a sub-ranking produces a gradient that deviates from typical training values, it is likely incorrect, and we accordingly reduce its weight. First, we calculate the absolute gradient value $g_{v} = |\nabla_{\pi_{u}(v)} \mathcal{L}_{rank}|$ for all sub-rankings and identify the largest gradient value $max(g_{v})$. We then divide the gradients into ten groups ${G_1,G_2, \ldots, G_{10}}$ according to the interval  $[0, max(g_{v})]$. For a particular gradient $g_v$, if it belongs to group $G_v$, we calculate the number $N_{v}$ of gradients in that group and the ratio of the sum of gradients across all groups:
\begin{align}
\alpha_{v} = \frac{N_v}{\sum_{y=0}^{10}N_y}.
\end{align}
This approach ensures that if the number of gradients in a group is small, indicating that they are outliers, $\alpha_{v}$ is reduced, aligning with our objective. This concept can be interpreted as a form of statistical gradient density. Given that the gradient is inherently a continuous quantity, direct calculation of the density might yield $1/N$ for each gradient density. To address this, we approximate the density by grouping the gradients. The final loss function is thus formulated as follows:
\begin{align}
\mathcal{L} = \mathcal{L}_{rank}^{\alpha}(u,\pi_u) + \beta \mathcal{L}_p,
\end{align}
where $\mathcal{L}_{rank}(u,\pi_u)$ processes the ranking information from the ranker. $\mathcal{L}_p$ aids the ranker in learning ranking knowledge. $\beta$ is a hyperparameter to control the relative weights.

\subsection{Discussion}
Our newly proposed loss function offers a more effective approximation of the ideal CF. This prompts further reflection: \textit{is there an intuitive connection between our loss function and the classic loss?}
\begin{theorem}
\label{th:3}
BPR is a special case of ranking loss where $k=2$. BPR with hard negative sampling is a special case of longest sub-ranking loss where $k>2$. 
\end{theorem}
\begin{proof}
\begin{align}
\mathcal{L}_{rank} &= -\sum_{v=1}^{k-1} \log\left(\frac{1}{\exp(s_u(\pi_{u}(v)) - s_u(\pi_{u}(v+1))) + 1}\right)\\
&= -\sum_{v=1}^{k-1} \log(\sigma(s_u(\pi_{u}(v)) - s_u(\pi_{u}(v+1)))) \\
&= -\sum_{v=1}^{k-1} \log(\sigma(s_u(\pi_{u}(v)) - \max_{w \neq v}\{s_u(\pi_{u}(w))\})).
\end{align}
It shows that if we set $k=2$, our method is equivalent to BPR, where $\pi_{u}(v)$ is a positive sample, and $\max_{w \neq v}\{s_u(\pi_{u}(w))\}$ is equal to a negative sample. When $k>2$, obviously, BPR only considers a certain longest sub-ranking, where $\pi_{u}(v)$ is still a positive sample, $\max_{w \neq v}\{s_u(\pi_{u}(w))\}$ is equivalent to performing \textbf{hard negative sampling} from $M - \{\pi_{u}(v)\}$ and selecting the hardest sample. 

\end{proof}

\section{Experiments}
\begin{table*}[t]
\begin{center}
 \caption{The performance comparison between PRP and other methods. The best results are boldfaced, and the improvements are significant over the BPR under a two-sided t-test with $p<0.05$.}
\label{tab:Overall-comprehension}
\begin{tabular}{@{}llccccccc@{}}
\toprule
\textbf{Dataset} & \textbf{Methods} & HR@10 & HR@20 & Recall@10 & Recall@20 & NDCG@10 & NDCG@20  \\ \midrule
\multirow{8}{*}{ML-1M} 
  & \textbf{BPR} & 0.7419 & 0.8374 & 0.1635 & 0.2502 & 0.2554 &  0.2619 \\
  & \textbf{SRNS} & 0.7611 & 0.8533 & 0.1746 & 0.2618 & 0.2639 &  0.2687 \\
 & \textbf{MixGCF} & 0.7601 & 0.8483 & 0.1739 & 0.2604 & 0.2657 & 0.2702  \\
 & \textbf{UIB} & 0.7346 & 0.8318 & 0.1559 & 0.2406 & 0.2453 & 0.2520 \\ 
& \textbf{SimpleX} & 0.7586 & 0.8518 & 0.1709 & 0.2573 & 0.2626 & 0.2675  \\ 
& \textbf{ANS} & 0.7675 & 0.8583 & 0.1759 & 0.2625 & 0.2624 & 0.2683  \\
& \textbf{PRP} & \textbf{0.7796}& \textbf{0.8657} & \textbf{0.1866} & \textbf{0.2765} & \textbf{0.2822 }& \textbf{0.2874} \\ 
\cmidrule{2-8}
& \textbf{Improvement} & \textbf{5.08\%}& \textbf{3.38\%} & \textbf{14.13\%} & \textbf{10.51\%} & \textbf{10.49\%}& \textbf{9.74\%} \\ 

 \midrule
 
\multirow{8}{*}{Gowalla} 
  & \textbf{BPR} & 0.1989 & 0.2763 & 0.0938 & 0.1389 & 0.0674 &  0.0804 \\
  & \textbf{SRNS} & 0.2295 & 0.3177 & 0.1171 & 0.1711 & 0.0814 & 0.0972  \\ 
 & \textbf{MixGCF} & 0.2401 & 0.3254 & 0.1189 & 0.1728 & 0.0853 & 0.1008  \\
 & \textbf{UIB} & 0.2275 & 0.3161 & 0.1140 & 0.1695 & 0.0794 & 0.0954 \\ 
& \textbf{SimpleX} & 0.2192 & 0.3082 & 0.1074 & 0.1617 & 0.0736 & 0.0892  \\
& \textbf{ANS} & 0.2467 & 0.3317 & 0.1206 & 0.1756 & 0.0869 & 0.1027  \\ 
& \textbf{PRP} & \textbf{0.2528} & \textbf{0.3395} & \textbf{0.1229} & \textbf{0.1777} & \textbf{0.0884} &\textbf{0.1040}  \\ 
\cmidrule{2-8}
& \textbf{Improvement} & \textbf{27.10\%}& \textbf{22.87\%} & \textbf{31.02\%} & \textbf{27.93\%} & \textbf{31.16\% }& \textbf{29.35\%} \\ 

 \midrule

 \multirow{8}{*}{Foursquare} 
   & \textbf{BPR} & 0.1717 & 0.2336 & 0.0262 & 0.0381 & 0.0283 &  0.0328 \\
  & \textbf{SRNS} & 0.1948 & 0.2596 & 0.0314 & 0.0418 & 0.0351 & 0.0392\\
 & \textbf{MixGCF} & 0.1994 & 0.2752 & 0.0316 & 0.0458 & 0.0351 & 0.0405  \\
 & \textbf{UIB} & 0.1662 & 0.2299 & 0.0265 & 0.0375 & 0.0280 & 0.0322 \\ 
& \textbf{SimpleX} & 0.1801 & 0.2650 & 0.0267 & 0.0421 & 0.0267 & 0.0329  \\ 
& \textbf{ANS} & 0.1958 & \textbf{0.2798} & 0.0303 & 0.0461 & 0.0321 & 0.0386  \\ 
& \textbf{PRP} & \textbf{0.2124} & 0.2752 & \textbf{0.0332} & \textbf{0.0461} & \textbf{0.0367} & \textbf{0.0413} \\ 
\cmidrule{2-8}
& \textbf{Improvement} & \textbf{23.70\%}& \textbf{17.81\%} & \textbf{26.72\%} & \textbf{20.73\%} & \textbf{29.68\%}& \textbf{25.91\%} \\ 
 
 \midrule
 
\multirow{8}{*}{Yelp} 
   & \textbf{BPR} & 0.1537 & 0.2398 & 0.0427 & 0.0717 & 0.0334 &  0.0432 \\
  & \textbf{SRNS} & 0.1834 & 0.2779 & 0.0520 & 0.0862 & 0.0406 & 0.0522 \\
 & \textbf{MixGCF} & 0.1891 & 0.2781 & 0.0538 & 0.0871 & 0.0424 & 0.0536  \\
 & \textbf{UIB} & 0.1592 & 0.2504 & 0.0445 & 0.0756 & 0.0350  & 0.0455 \\ 
& \textbf{SimpleX} & 0.1850 & 0.2780 & 0.0519 & 0.0861 & 0.0410 & 0.0524  \\ 
& \textbf{ANS} & 0.1844 & 0.2777 & 0.0524 & 0.0862 & 0.0414 & 0.0527  \\ 
& \textbf{PRP} & \textbf{0.2000} & \textbf{0.2940} & \textbf{0.0577} & \textbf{0.0928} & \textbf{0.0464} & \textbf{0.0581} \\ 
\cmidrule{2-8}
& \textbf{Improvement} & \textbf{30.12\%}& \textbf{22.60\%} & \textbf{35.13\%} & \textbf{29.43\%} & \textbf{38.92\%}& \textbf{34.49\%} \\ 

 \bottomrule
 \end{tabular}
  \end{center}
\end{table*}
In this section, we conduct comprehensive experiments to demonstrate the performance of PRP.
\subsection{Experimental Settings}
\textbf{Datasets}. We evaluate PRP on four widely used public real-world datasets: (1) \textbf{MovieLens} contains user ratings on movies. We use 1M versions and treat rating movies as interacted items. (2) \textbf{Yelp} contains user reviews of restaurants and bars after Jan.1st, 2018. (3) \textbf{Gowalla} is a check-in dataset where users share their locations. (4) \textbf{Foursquare} consists of check-ins in NYC and Tokyo over approximately 10 months. These datasets have different statistical properties, which can reliably validate the performance of a model~\cite{CCC22}.

\textbf{Baselines}. 
To assess the effectiveness of PRP, we compare it against several competitive methods that represent different research directions: \textbf{BPR}~\cite{RFG12}, \textbf{SimpleX}~\cite{MZW21},  \textbf{UIB}~\cite{ZZY22}, \textbf{SRNS}~\cite{DQY20},\textbf{MixGCF}~\cite{HDD21}, \textbf{ANS}~\cite{ZCL23}. 
Furthermore, to validate the broad applicability of PRP, we integrate it with various mainstream CF models, including: \textbf{MF}~\cite{RFG12}, \textbf{NGCF}~\cite{WHW19}, and \textbf{LightGCN}~\cite{HDW20}.

All data preprocessing, dataset partitioning and implementation of the methods are carried out using the RecBole v1.1.1 framework~\cite{ZMH21}.

\subsection{Overall Performance}
\begin{table*}[t]
\caption{Experimental results of different CF models with (w) or without (w/o) the PRP.}
\centering
\begin{tabular}{@{}llcccccccc@{}}
\toprule
\multirow{2}{*}{\textbf{Dataset}} & \multirow{2}{*}{\textbf{Metric}} & \multicolumn{2}{c}{\textbf{MF}} & \multicolumn{2}{c}{\textbf{LightGCN}} & \multicolumn{2}{c}{\textbf{NGCF}} & 
\\ \cmidrule(l){3-9} 
 &  & w/o & w & w/o & w & w/o & w  \\ \midrule
\multirow{6}{*}{Yelp} 
 & Recall@10 & 0.0427 & \textbf{0.0577} & 0.0543 & \textbf{0.0655} & 0.0439 & \textbf{0.0514}  \\
 & Recall@20 & 0.0717 & \textbf{0.0928} & 0.0884 & \textbf{0.1044} & 0.0738 & \textbf{0.0834}   \\
 & HR@10 & 0.1537 & \textbf{0.2000} & 0.1896 & \textbf{0.2210} & 0.1593 & \textbf{0.1811}   \\
 & HR@20 & 0.2398 & \textbf{0.2940} & 0.2847 & \textbf{0.3232} & 0.2460 & \textbf{0.2705}  \\
 & NDCG@10 & 0.0334 & \textbf{0.0464} & 0.0432 & \textbf{0.0525} & 0.0345 & \textbf{0.0406}  \\
 & NDCG@20 & 0.0432 & \textbf{0.0581} & 0.0546 & \textbf{0.0655} & 0.0447 & \textbf{0.0512}   \\
\cmidrule(l){2-9} 
  & Improvement &  & \textbf{31.78\%} &  & \textbf{18.38\%} &  & \textbf{14.33\%} \\ \midrule 

\multirow{6}{*}{Gowalla} 
 & Recall@10 & 0.0938 & \textbf{0.1236} & 0.1280 & \textbf{0.1360} & 0.1011 & \textbf{0.1225}   \\
 & Recall@20 & 0.1389 & \textbf{0.1787} & 0.1849 & \textbf{0.1967} & 0.1490 & \textbf{0.1745}   \\
 & HR@10 & 0.1989 & \textbf{0.2534} &  0.2553 & \textbf{0.2712} & 0.2110 & \textbf{0.2468}  \\
 & HR@20 & 0.2763 & \textbf{0.3366} & 0.3432 & \textbf{0.3626} & 0.2900 & \textbf{0.3321}   \\
 & NDCG@10 & 0.0674 & \textbf{0.0896} & 0.0917 & \textbf{0.0981} & 0.0722 & \textbf{0.0856} \\
 & NDCG@20 & 0.0804 & \textbf{0.1052} & 0.1079 & \textbf{0.1155} & 0.0859 &  \textbf{0.1010}  \\
\cmidrule(l){2-9} 
  & Improvement &  & \textbf{29.91\%} &  & \textbf{6.42\%} &  & \textbf{17.65\%} \\ \midrule 
\end{tabular}
\label{tab:with}
\end{table*}
Table~\ref{tab:Overall-comprehension} presents the performance results of PRP compared to several baseline methods. We also integrated PRP with various mainstream CF models, as shown in Table~\ref{tab:with}. The reported improvements are statistically significant under a two-sided t-test with $p < 0.05$. Our observations are as follows: 
(1) PRP consistently achieves state-of-the-art results across almost all cases. Notably, it significantly outperforms traditional methods like BPR, with the highest improvement of 38.9\% observed on the Yelp dataset. This underscores the substantial performance gains achieved by incorporating ranking information into the model.
(2) Negative sampling (NS) methods achieve sub-optimal results. While NS methods can capture information from sub-rankings where $k>2$, BPR is limited to $k=2$. This ability of NS methods to utilize more extensive sub-ranking information explains their improved performance over BPR. However, NS methods still overlook other sub-ranking information, leading to their lower performance compared to PRP.
(3)  Methods such as ANS and MixGCF generate harder synthetic negative samples $i_n^{syn}$, leading to impressive results. This is reasonable, as generating harder $i_n^{syn}$ implicitly aligns with the ranking information $i_n^{syn} > i_n$. However, synthesizing these negative samples requires careful design and involves significant computational overhead. Moreover, since only a specific ranking is considered, there remains a performance gap when compared to PRP. (4) Methods such as SimpleX use large numbers of negative samples to construct more pairwise interactions. Although they have the potential to recover ranking information if we have $\Omega(|U|N^2)$ accurate pairs~\cite{ABB94,MG17}, this is impractical in real-world scenarios~\cite{LSC21,LGZ24}. (5) As shown in Table~\ref{tab:with}, integrating PRP leads to significant performance improvements for all base models in all cases. Surprisingly, even with a simple BPR-MF model, PRP achieves performance comparable to more sophisticated GNN models. This demonstrates the superior effectiveness and broad applicability of the PRP approach.

\subsection{In-depth Analysis}
\begin{figure*}[t]
\center
\subfigure[Foursquare]{
\begin{minipage}[t]{0.235\linewidth} 
\centering
\includegraphics[width=\linewidth]{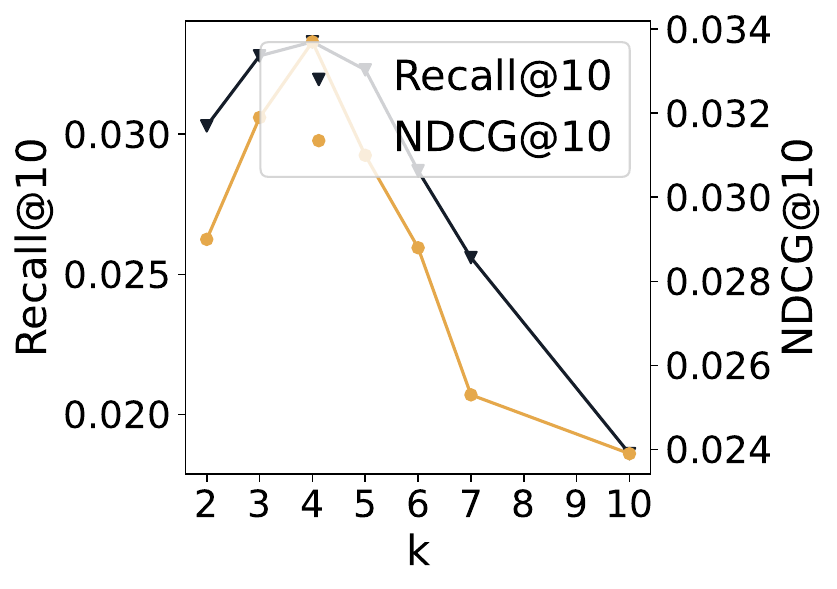}
\label{fig:k_four}
\end{minipage}
}
\subfigure[Gowalla]{
\begin{minipage}[t]{0.227\linewidth}
\centering
\includegraphics[width=\linewidth]{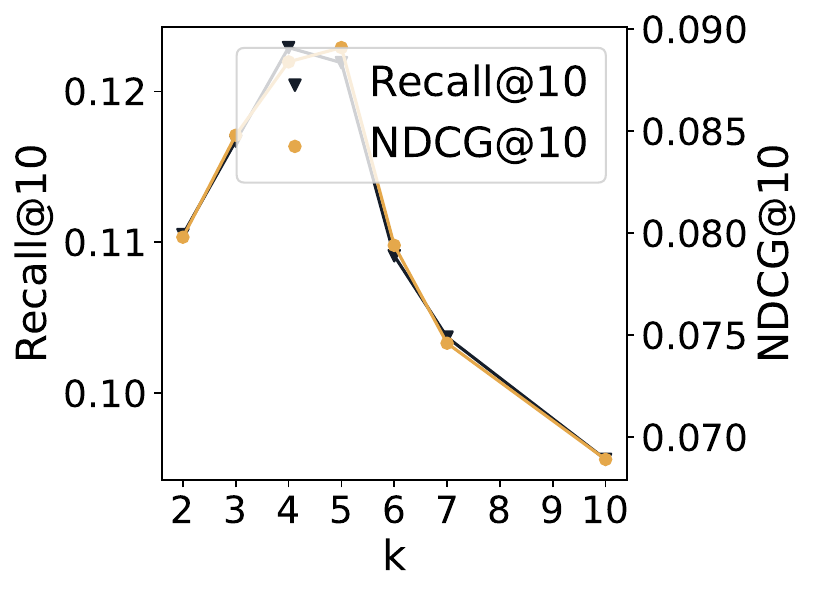}
\label{fig:k_Gowalla}

\end{minipage}
}
\subfigure[Foursquare]{
\begin{minipage}[t]{0.232\linewidth}
\centering
\includegraphics[width=\linewidth]{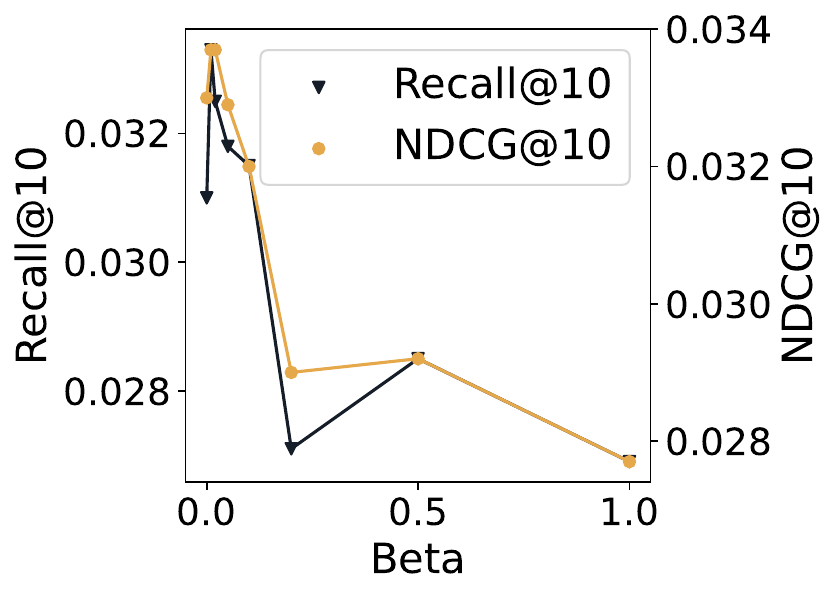}
\label{fig:beta_four}
\end{minipage}
}
\subfigure[Gowalla]{
\begin{minipage}[t]{0.238\linewidth}
\centering
\includegraphics[width=\linewidth]{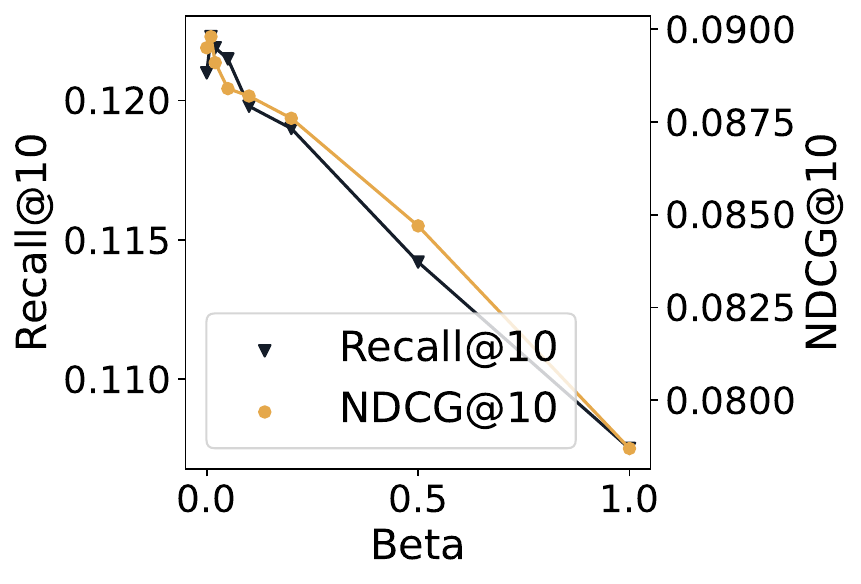}
\label{fig:beta_Gowalla}
\end{minipage}
}
\caption{Experiments to demonstrate the influence of ranking length and supervised magnitude.}
\vspace{-4mm}
\end{figure*}

\textbf{Ranking Length}. The parameter $k$ determines the number of items to be ranked. Intuitively, the larger the $k$, the closer it approximates the full ranking needed for ideal CF, potentially enhancing performance. However, as $k$ increases, ranking becomes more challenging. Due to data quality limitations, model capacity, and inherent noise, it is impractical to let $k$ increase indefinitely. Experimental results presented in Figure~\ref{fig:k_four}, ~\ref{fig:k_Gowalla} confirm this, showing a sharp performance drop for $k>5$, indicating the ranker's inability to accurately rank and the detrimental impact of incorrect rankings on model performance. Additionally, a larger $k$ affects the method's efficiency, necessitating careful selection based on the specific application scenario.

\textbf{Supervised Magnitude}. The parameter $\beta$ adjusts the supervised magnitude of $L_p$ on the ranker. Experimental results, presented in Figure~\ref{fig:beta_four}, ~\ref{fig:beta_Gowalla}, show that performance initially increases with $\beta$ but decreases beyond a certain point. Initially, a higher $\beta$ allows $L_p$ to provide more information for ranker training, improving pseudo-ranking quality. However, an excessively high $\beta$ simplifies the noise injection ranking problem too much, making $L_p$ optimization easier than the recommendation task. This shifts the model's focus away from the recommendation task, degrading performance.
\begin{figure*}[t]
\center
\subfigure[Foursquare]{
\begin{minipage}[t]{0.256\linewidth} 
\centering
\includegraphics[width=\linewidth]{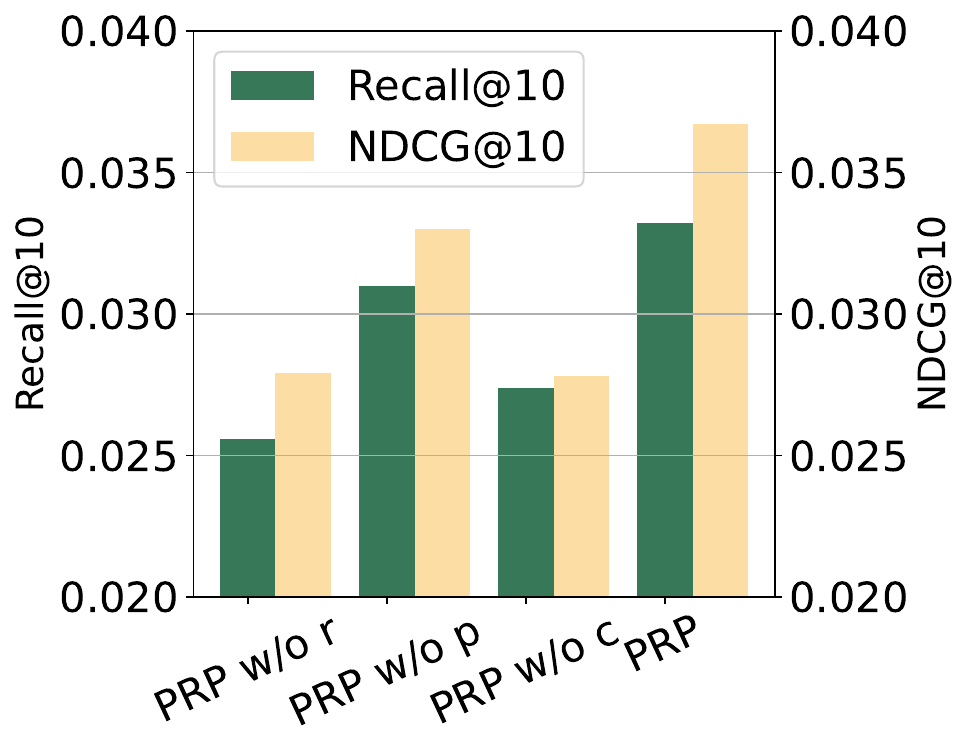}
\label{fig:ablation_foursquare}
\end{minipage}
}
\subfigure[Gowalla]{
\begin{minipage}[t]{0.242\linewidth}
\centering
\includegraphics[width=\linewidth]{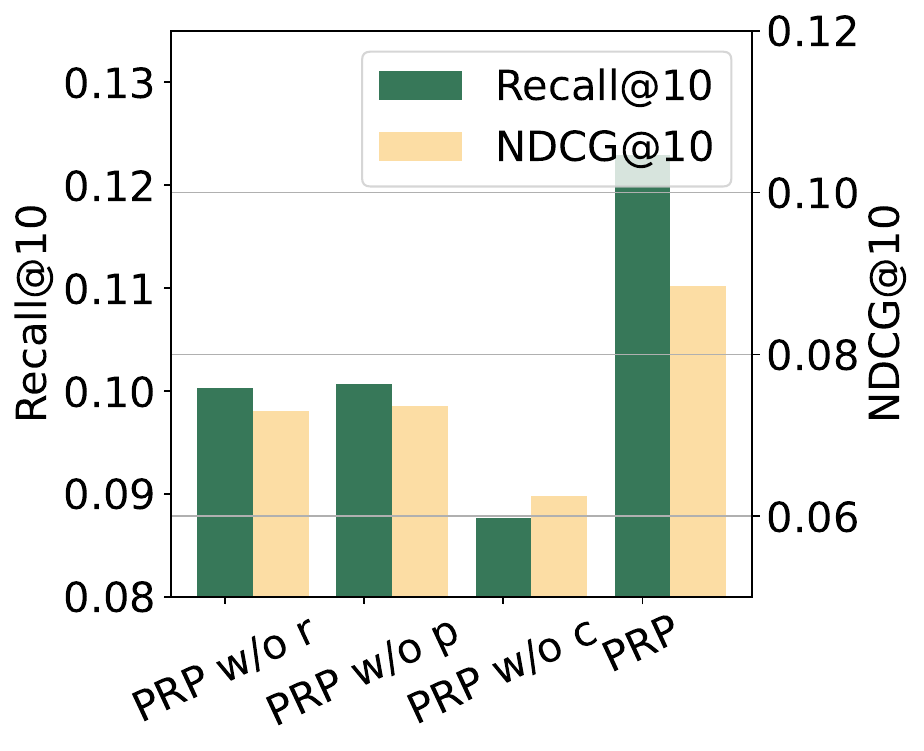}
\label{fig:ablation_Gowalla}
\end{minipage}
}
\subfigure[Yelp]{
\begin{minipage}[t]{0.207\linewidth}
\centering
\includegraphics[width=\linewidth]{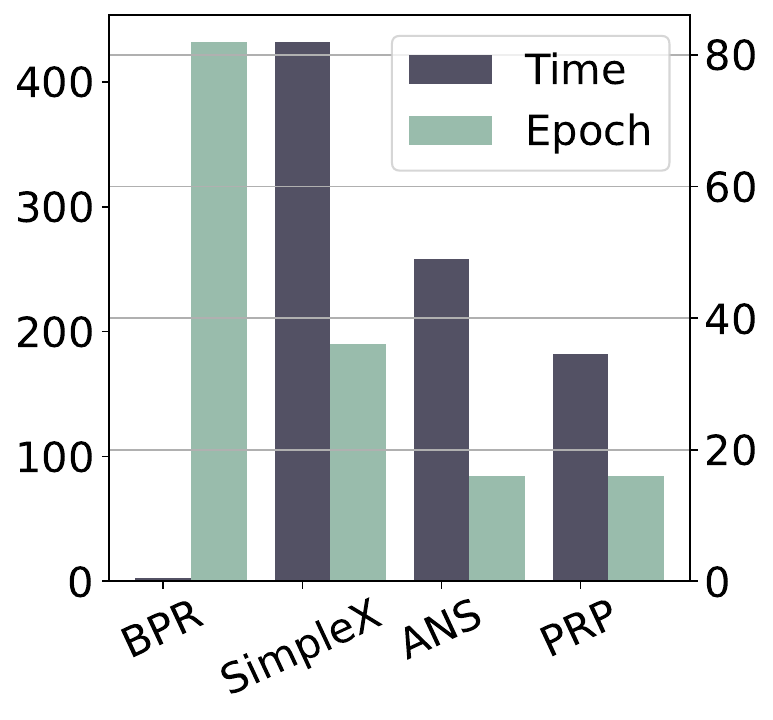}
\label{fig:time_Yelp}
\end{minipage}
}
\subfigure[Gowalla]{
\begin{minipage}[t]{0.207\linewidth}
\centering
\includegraphics[width=\linewidth]{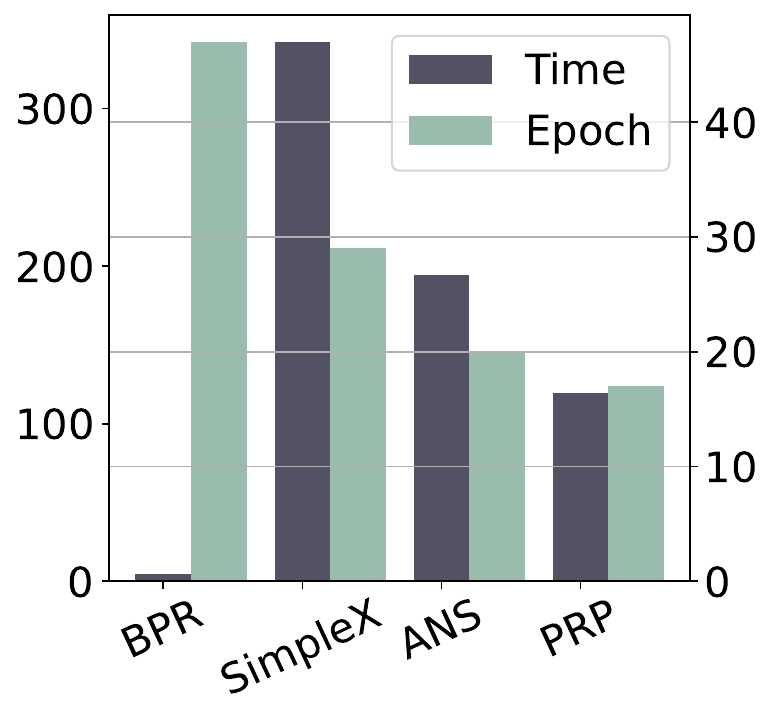}
\label{fig:time_Gowalla}
\end{minipage}
}
\caption{Experimental results to demonstrate the ablation study and efficiency analysis.}
\label{fig:trap}
\end{figure*}

\textbf{Ablation Study}. We analyze the effectiveness of different components through the following variants: (1) PRP without ranker (PNN w/o r), (2) PRP without $\mathcal{L}_p$ (PNN w/o p), (3) PRP without confidence (PNN w/o c). The results, presented in Figure~\ref{fig:ablation_foursquare}, ~\ref{fig:ablation_Gowalla} indicate that each component positively contributes to model performance. In essence, these three components ensure that the model effectively utilizes accurate ranking information, underscoring the importance of constructing a reasonable ranking. If the ranking is incorrect, the ranking loss becomes meaningless.

\textbf{Efficiency Analysis}. Consistent with previous works~\cite{MZW21, WYM22, ZCL23}, we compare the efficiency of PRP with other representative methods (BPR, SimpleX, and ANS). Figure~\ref{fig:time_Yelp},~\ref{fig:time_Gowalla} present the average training time (sec.) per epoch and the number of epochs required to converge. As anticipated, BPR is the most efficient due to its simplicity. SimpleX, on the other hand, shows the least efficiency, attributed to its dependence on massive negative samples and filtering of uninformative ones. PRP emerges as the second most efficient method. Given its outstanding performance, we assert that PRP stands out as a highly desirable choice.

\section{Related Work}
Pairwise loss functions lie in the core of existing implicit CF. Since they aim to maximize the posterior probability $\prod_ {(u, i_p, i_n) \in D}p(\Theta \mid i_p>_{u}i_n)$, there have been two major research directions. The first line focuses on how to effectively form pairwise relationships (i.e., $>_{u}$). BPR~\cite{RFG12, SGZ20,YYG21} assumes that an interacted item is more similar to the user than a randomly sampled uninteracted item. CML~\cite{HYC17} encodes not only users' preferences but also the user-user and item-item similarity. SimpleX~\cite{MZW21} introduces massive negative samples and removes uninformative items. UIB~\cite{ZZY22} introduces a user interest boundary to penalize samples crossing a threshold. The second line of research aims to find better negative items to contrast positive items. PNN~\cite{ZCH24} introduces a neutral class and a novel positive-neutral-negative learning paradigm. The state-of-the-art focuses on identifying hard negative samples that help tighten a model's decision boundary. DNS~\cite{ZCW13} selects negative items that are more similar to a user. IRGAN~\cite{WYZ17} utilizes a generative adversarial network to compute the probabilities of negative samples by a min-max game. ReinforcedNS~\cite{DQH19} uses reinforcement learning to guide the sampling process. MixGCF~\cite{HDD21} integrates information from a graph structure and positive samples to enhance the hardness of negative items. GDNS~\cite{ZZH22} develops a gain-aware function to calculate the probability of an item being a real negative sample. SRNS~\cite{DQY20} favors low-variance samples to avoid false negative samples. DENS~\cite{LCZ23} disentangles relevant and irrelevant factors of samples to select appropriate negative samples. ANS~\cite{ZCL23} proposes to generate synthetic negative samples to improve performance. 

\section{Conclusion}
In this work, we creatively apply the concept of multiple ordinal classification to highlight the inherent gap between pairwise methods and the ideal CF. Based on this theory, we introduce an innovative pseudo-ranking paradigm (PRP). Our implementation features a novel ranker that generates pseudo-rankings, supervised by a noise mechanism. We also propose a new loss function to handle ranking information effectively. Furthermore, this loss is enhanced with a gradient-based confidence mechanism that mitigates abnormal rankings. Extensive experiments demonstrate that PRP significantly outperforms state-of-the-art methods. This work marks a significant advancement in CF modeling, shifting the focus from pairwise comparisons to a more comprehensive ranking approach, and represents a major step toward realizing ideal CF models.

\clearpage
\section*{Acknowledgements}
This work was supported by the Heilongjiang Key R\&D Program of China under Grant No. GA23A915 and partially supported by Hong Kong Baptist University IG-FNRA project (RC-FNRA-IG/21-22/SCI/01) and Key Research Partnership Scheme (KRPS/23-24/02).
\bibliography{aaai25}

\end{document}